\documentclass[11pt]{article}
\usepackage[T1]{fontenc}
\usepackage{graphicx}

\usepackage{ textcomp }

\usepackage[hyphens]{url}
\usepackage[colorlinks=true]{hyperref}
\usepackage[ocgcolorlinks]{ocgx2} 
\hypersetup{urlcolor=blue, citecolor=red, linkcolor=blue, breaklinks=true}
\usepackage[hyphenbreaks]{breakurl}
\usepackage[normalem]{ulem}

\usepackage{amsmath,amssymb}
\usepackage{bm}
\usepackage{enumerate}
\usepackage{amsthm}
\usepackage[all,2cell]{xy}
\usepackage{mathrsfs}
\usepackage{mathtools}
\mathtoolsset{showonlyrefs}
\usepackage{tikz-cd}

\usepackage{physics}

\usepackage{caption}
\usepackage{subcaption}

\usetikzlibrary{cd}

\usepackage[a4paper, left=25mm, right=25mm, top=25mm, bottom=25mm]{geometry}

\usepackage{marginnote}

\usepackage{imakeidx}
\makeindex[intoc]

\usepackage[nottoc]{tocbibind}

\usepackage[colorinlistoftodos]{todonotes} 

\newtheorem{thm}{Theorem}[section]
\newtheorem{dfn}[thm]{Definition}
\newtheorem{prop}[thm]{Proposition}

\newtheorem{lem}[thm]{Lemma}

\newtheorem{exmpl}[thm]{Example}
\newtheorem{rmrk}[thm]{Remark}

\newcommand\restr[2]{{
  \left.\kern-\nulldelimiterspace 
  #1 
  \right|_{#2} 
}}

\newcommand{\R}{\mathbb{R}}
\renewcommand{\d}{\mathrm{d}}
\let\dd\d

\newcommand{\Cinfty}{\mathscr{C}^\infty}
\newcommand{\T}{\mathrm{T}}
\newcommand{\cT}{\mathrm{T}^\ast}
\newcommand{\Id}{\mathrm{Id}}

\newcommand{\Lie}{\mathscr{L}}
\newcommand{\liedv}[1]{\Lie_{#1}}
\newcommand*{\contr}[1]{\iota_{#1}}

\newcommand{\X}{\mathfrak{X}}

\newcommand{\Reeb}{R}
\newcommand{\Rt}{R_t}
\newcommand{\Rz}{R_z}
\let\Rs\Rz

\newcommand{\parder}[2]{\frac{\partial #1}{\partial #2}}

\newcommand{\tparder}[2]{\partial #1/\partial #2}

\DeclareMathOperator{\Ima}{Im}

\DeclareMathOperator{\pr}{pr}
\DeclareMathAlphabet{\mathpzc}{OT1}{pzc}{m}{it}
\def\d{\mathrm{d}}

\let\oldemph\emph
\let\emph\textbf

\usepackage[style=ext-numeric,citestyle=numeric-comp,sorting=nyt,sortcites, backend=biber, giveninits=true, maxnames=99]{biblatex}
\bibliography{biblio.bib}
\DeclareFieldFormat[article,periodical]{volume}{\mkbibbold{#1}}
\DeclareFieldFormat[article,periodical]{number}{\mkbibparens{#1}}
\renewbibmacro{in:}{%
  \ifentrytype{article}
    {}
    {\printtext{\bibstring{in}\intitlepunct}}}

\DeclareFieldFormat{issuedate}{#1}
\renewcommand{\jourvoldelim}{\addcomma\space}

\renewbibmacro*{journal+issuetitle}{%
  \usebibmacro{journal}%
  \setunit*{\jourvoldelim}%
  \iffieldundef{series}
    {}
    {\setunit*{\jourserdelim}%
     \printfield{series}%
     \setunit{\servoldelim}}%
  \usebibmacro{volume+number+eid}%
  \setunit{\addcolon\space}%
  \usebibmacro{issue}%
  \setunit{\bibpagespunct}%
  \printfield{pages}%
  \setunit{\space}%
  \usebibmacro{issue+date}%
  \newunit}

\renewbibmacro*{note+pages}{%
  \printfield{note}%
  \newunit}

\DeclareFieldFormat[article,periodical]{date}{\mkbibparens{#1}}

\AtEveryBibitem{
  \ifentrytype{online}{}{
    \ifentrytype{thesis}{}{
        \clearfield{url}
        \clearfield{urldate}
      }
  }
  \ifentrytype{unpublished}{}{ 
  \clearfield{note}
  }
  \clearfield{language}
}

\DeclareSourcemap{
  \maps[datatype=bibtex]{
    \map[overwrite=true]{
      \step[fieldset=urldate, null]
      \step[fieldset=language, null]
      \step[fieldset=address, null]
     \step[fieldset=pagetotal, null]
    }
  }
}

\title{{\sffamily Hamilton--Jacobi theory and integrability for\\ autonomous and non-autonomous contact systems}}

\author{{\sffamily 
$^{a,b}$Manuel de León%
\thanks{e-mail:
   mdeleon@icmat.es \ ORCID: 0000-0002-8028-2348}\ ,\
$^a$Manuel Lainz%
\thanks{e-mail:
   manuel.lainz@icmat.es \ ORCID: 0000-0002-2368-5853}\ ,\
$^a$Asier López-Gordón%
\thanks{e-mail:
   asier.lopez@icmat.es \ ORCID: 0000-0002-9620-9647}\ ,\
$^c$Xavier Rivas%
\thanks{e-mail:
   xavier.rivas@unir.net \ ORCID: 0000-0002-4175-5157}\ ,\
}
\\[1ex]
\normalsize\itshape\sffamily
$^a$Instituto de Ciencias Matem\'aticas,
Consejo Superior de Investigaciones Cient\'ificas
\\[1ex]
\normalsize\itshape\sffamily
Calle Nicolás Cabrera 13-15, 28049, Madrid, Spain.
\\[1ex]
\normalsize\itshape\sffamily
$^b$Real Academia de Ciencias Exactas, Físicas y Naturales, Madrid, Spain.
\\[1ex]
\normalsize\itshape\sffamily
$^c$Escuela Superior de Ingenier\'{\i}a y Tecnolog\'{\i}a,
\normalsize\itshape\sffamily
Universidad Internacional de La Rioja, Logro\~no, Spain.
\\[1ex]
}

\date{{\sffamily \today}}

\begin{document}

\maketitle

\begin{abstract}\noindent
    In this paper, we study the integrability of contact Hamiltonian systems, both time-dependent and independent. In order to do so, we construct a Hamilton--Jacobi theory for these systems following two approaches, obtaining two different Hamilton--Jacobi equations. Compared to conservative Hamiltonian systems, contact Hamiltonian systems depend of one additional parameter.
    The fact of obtaining two equations reflects whether we are looking for solutions depending on this additional parameter or not. In order to illustrate the theory developed in this paper, we study three examples: the free particle with a linear external force, the freely falling particle with linear dissipation and the damped and forced harmonic oscillator.
\end{abstract}


\noindent\textbf{MSC\,2020 codes:}
37J55, 
70H20; 
70H33, 
 53D10, 
53Z05 
{\setcounter{tocdepth}{2}
\def\baselinestretch{1}
\small
\def\addvspace#1{\vskip 1pt}
\parskip 0pt plus 0.1mm
\tableofcontents
}

\newpage



\section{Introduction}

Recently there has been a renewed interest in using contact geometry \cite{Gei2008, Kholodenko2013} to describe mechanical systems. These systems, unlike symplectic Hamiltonian systems, lead to dissipated rather than conserved quantities \cite{deLeon2019a,deLeon2020a,Gaset2020a,Gaset2022}.
These systems are also relevant to describe
mechanical systems with certain types of damping \cite{deLeon2019a,Lainz2022,Bravetti2017,Bravetti2017a}, quantum mechanics \cite{Ciaglia2018,Budiyono2012}, Lie systems \cite{deLucas2022}, circuit theory \cite{Got2016}, thermodynamics \cite{Simoes2020,Bra2018}, control theory \cite{Maschke2018, vanderSchaft2018,deLeon2020b} and black holes \cite{Ghosh2019}, among many others \cite{Bravetti2020,Grabowska2022a}. 
The underlying variational principle is the so-called Herglotz principle \cite{Herglotz1930,deLeon2019b,Lopez-Gordon2022}, a generalization of the well-known Hamilton principle, which gives rise to action-dependent Lagrangian systems. These Lagrangians are becoming popular in theoretical physics \cite{Lazo2017, Lazo2018, Lazo2019}. Recently, contact mechanics have been generalized to deal with classical field theories with damping \cite{deLeon2022c,Gaset2020,Gaset2021a,Gracia2022,Rivas2022a}

Hamilton--Jacobi theory provides a remarkably powerful method to integrate
the dynamics of many Hamiltonian systems. In particular, for a completely integrable system, if one knows a complete solution of the Hamilton--Jacobi problem, the dynamics of the system can be reduced to quadratures \cite{Grillo2016, Grillo2021a, Grillo2021,Goldstein1980,Arnold1978}. Geometrically, the Hamilton--Jacobi problem consists on finding a section $\gamma$ of $\pi_Q\colon \cT Q \to Q$ which transforms integral curves of a projected vector field $X_H^\gamma$ on $Q$ into integral curves of the dynamical vector field $X_H$ on $\cT Q$ \cite{Abraham2008, Carinena2006}. This idea can be naturally extended to other vector bundles. 
As a matter of fact, it has been applied in many other different contexts, such as nonholonomic systems \cite{Carinena2010,Iglesias-Ponte2008, deLeon2010, Ohsawa2011a}, singular Lagrangian systems \cite{deLeon2012, Leok2012a, deLeon2013}, 
higher-order systems \cite{Colombo2014}, field theories \cite{Campos2015, deLeon2014a, deLeon2020, Vitagliano2012, Zatloukal2016} or systems with external forces \cite{deLeon2022, deLeon2022b}.\
 A unifying Hamilton--Jacobi theory for almost-Poisson manifolds was developed
in \cite{deLeon2014}. Hamilton--Jacobi theory has also been extended to Hamiltonian systems with non-canonical symplectic structures \cite{Martinez-Merino2006}, non-Hamiltonian systems \cite{Rashkovskiy2020}, locally conformally symplectic manifolds \cite{Esen2021}, Nambu--Poisson \cite{deLeon2017a} and Nambu--Jacobi \cite{deLeon2019} manifolds, Lie algebroids \cite{Leok2012} and implicit differential systems \cite{Esen2018,Esen2021a}. 
The applications of Hamilton--Jacobi theory include the relation between classical and quantum mechanics \cite{Budiyono2012,Carinena2009, Marmo2009}, information geometry \cite{Ciaglia2017, Ciaglia2017a}, control theory \cite{Sakamoto2002} and the study of phase transitions \cite{Kraaij2021}. Hamilton--Jacobi theory for autonomous contact Hamiltonian systems has been studied in \cite{Grabowska2022,deLeon2021d,deLeon2017}.

We have recently initiated the study of time-dependent contact Hamiltonian systems \cite{deLeon2022d, Rivas2023, Gaset2022}, and the underlying geometric structures, which we call cocontact manifolds since they are a combination of cosymplectic (the setting for studying time-dependent Hamiltonian systems) and contact structures. Such structures consist of two one-forms, $\tau$ and $\eta$, where $\tau$ is closed and $\tau \wedge \eta \wedge (\d\eta)^n$ is a volume form, in a $(2n+2)$-dimensional manifold. The local model for cocontact manifolds is the product bundle $\R \times \cT Q \times \R$ with a cocontact structure induced by the canonical symplectic structure of the cotangent bundle. In fact, in \cite{deLeon2022d} we have been able to identify that a cocontact structure gives rise to a Jacobi structure whose characteristic foliation is formed by contact leaves.

The aim of the present paper is to develop a Hamilton--Jacobi theory for time-dependent contact Hamiltonian systems. This will also allow us to construct time-dependent solutions of the Hamilton--Jacobi problem for autonomous contact systems, which, unlike time-independent solutions, cover nonzero energy levels.
We follow the line undertaken in previous papers \cite{deLeon2021d, deLeon2017}, considering sections of the canonical fibrations $\R \times \cT Q \times \R \to \R \times Q$ and $\R \times \cT Q \times \R \to \R \times Q\times  \R$, which allows us to project the Hamiltonian vector field to the base and, by comparing the values on the section, we obtain the corresponding Hamilton--Jacobi equations. This study is particularly useful since it allows us to study the symmetries, conserved quantities and integrability of the system.

In the first of the approaches, where sections of $\R \times \cT Q \times \R \to \R \times Q$ are considered, complete solutions depend on $n+1$ parameters (instead of the usual $n=\dim Q$ parameters in the classical Hamilton--Jacobi theory). We also make use of this approach to construct complete solutions, depending on $n$ parameters, for autonomous contact Hamiltonian systems. In the second approach we consider sections of $\R \times \cT Q \times \R \to \R \times Q\times \R$, and complete solutions depend of $n$ parameters (roughly speaking, the additional parameter is absorbed by the extra $\R$-component of the base). Furthermore, this second approach motivates a new definition of integrable contact Hamiltonian system. 

The paper is structured as follows. Section \ref{section_review} is devoted to review time-dependent contact Hamiltonian systems introducing the basic elements needed. In Section \ref{section_symmetries} we study symmetries and dissipated quantities in cocontact Hamiltonian systems. In Section \ref{section_HJ_action_indep} we develop the action-independent approach to the Hamilton--Jacobi problem, study complete solutions and apply our results for integrating time-independent contact Hamiltonian systems. We also present an example: a free particle with linear friction.
In Section \ref{section_HJ_action_dep} we deal with the action-dependent approach to the Hamilton--Jacobi problem, study complete solutions and introduce a new definition of integrable contact system. We also discuss two examples as applications of this approach: the freely falling particle with linear dissipation and the damped and forced harmonic oscillator.

From now on, all the manifolds and mappings are assumed to be smooth, connected and second-countable. Sum over crossed repeated indices is understood.



\section{Review on time-dependent contact systems} \label{section_review}

In this section we are going to review some fundamentals on cocontact geometry and time-dependent contact Hamiltonian systems (for more details see \cite{deLeon2022d}).

\subsection{Cocontact manifolds}

\begin{dfn}\label{dfn:cocontact-manifold}
    A \emph{cocontact structure} on a $(2n+2)$-dimensional manifold $M$ is a couple $(\tau,\eta)$, where $\tau,\eta\in\Omega^1(M)$ and $\d\tau = 0$, such that $\tau\wedge\eta\wedge(\d\eta)^n$ is a volume form on $M$. In this case, $(M,\tau,\eta)$ is called a \emph{cocontact manifold}.
\end{dfn}

Given a cocontact manifold $(M,\tau,\eta)$, the distribution $\mathcal{H} = \ker\eta$ is called the \emph{horizontal} or \emph{contact distribution}. Notice that this distribution has corank one and is maximally non-integrable.

\begin{exmpl}\label{exmpl:R-contact}\rm
	Let $(P,\eta_0)$ be a contact manifold\footnote{A \textbf{contact structure} on an odd-dimensional manifold $M$ is a one-codimensional maximally non-integrable distribution $C$ on $M$. In this case, $(M,C)$ is a \textbf{contact manifold}. A \textbf{contact form} on $M$ is a one-form $\eta\in\Omega^1(M)$ such that $\ker\eta$ becomes a contact structure on $M$. In this case, $(M,\eta)$ is called a \textbf{co-oriented contact manifold} \cite{Gei2008}. However, since we are only interested in local aspects of contact manifolds, we will consider that all our contact manifolds are co-oriented.} and consider the product manifold $M = \R\times P$. Denoting by $\d t$ the pullback to $M$ of the volume form in $\R$ and by $\eta$ the pullback of $\eta_0$ to $M$, we have that $(\d t, \eta)$ is a cocontact structure on~$M$.
\end{exmpl}

\begin{exmpl}\rm
	Let $(P,\tau,-\d\theta)$ be an exact cosymplectic manifold \cite{Cantrijn1992} and consider the product manifold $M = P\times\R$. Denoting by $z$ the coordinate in $\R$ we define the one-form $\eta = \d z - \theta$. Then, $(\tau, \eta)$ is a cocontact structure on $M = P\times\R$.
\end{exmpl}

\begin{exmpl}[Canonical cocontact manifold]\label{ex:canonical-cocontact-manifold}\rm
	Let $Q$ be an $n$-dimensional smooth manifold with local coordinates $(q^i)$ and consider its cotangent bundle $\cT Q$ with induced natural coordinates $(q^i, p_i)$. Consider the product manifolds $\R\times\cT Q$ with coordinates $(t, q^i, p_i)$, $\cT Q\times\R$ with coordinates $(q^i, p_i, z)$ and $\R\times\cT Q\times\R$ with coordinates $(t, q^i, p_i, z)$ and the canonical projections
	\begin{center}
		\begin{tikzcd}
			& \R\times\cT Q\times\R \arrow[dl, swap, "\rho_1"] \arrow[dr, "\rho_2"] \arrow[dd, "\pi"] & \\
			\R\times\cT Q \arrow[dr, swap, "\pi_2"] & & \cT Q\times\R \arrow[dl, "\pi_1"] \\
			& \cT Q &
		\end{tikzcd}
	\end{center}
	Let $\theta_0\in\Omega^1(\cT Q)$ be the Liouville one-form of the cotangent bundle, which has local expression $\theta_0 = p_i\d q^i$. Then, $(\d t, \theta_2)$, {where $\theta_2 = \pi_2^\ast\theta_0$}, is a cosymplectic structure in $\R\times\cT Q$. On the other hand, if $\theta_1 = \pi_1^\ast\theta_0$, we have that $\eta_1 = \d z - \theta_1$ is a contact form in $\cT Q\times\R$.
	
	Finally, consider the 1-form $\theta = \rho_1^\ast\theta_2 = \rho_2^\ast\theta_1 = \pi^\ast\theta_0\in\Omega^1(\R\times\cT Q\times\R)$ and let $\eta = \d z - \theta$. Then, $(\d t, \eta)$ is a cocontact structure in $\R\times\cT Q\times\R$. The local expression of the one-form $\eta$ is
	$$ \eta = \d z - p_i\d q^i\,. $$
\end{exmpl}

Given a cocontact manifold $(M, \tau, \eta)$, we have the \emph{flat isomorphism}.
\begin{equation}
    \flat: v\in\T M \longmapsto (\contr{v}\tau)\tau + \contr{v}\d\eta + \left(\contr{v}\eta\right)\eta\in\cT M\, .
\end{equation}
This isomorphism can be trivially extended to an isomorphism of $\Cinfty(M)$-modules $\flat: \X(M)\to\Omega^1(M)$. The inverse of the flat isomorphism is denoted by $\sharp = \flat^{-1}\colon\Omega^1(M)\to\X(M)$ and called the \emph{sharp isomorphism}.

Moreover, we have the following results, whose proofs can be found in \cite{deLeon2022d}.
\begin{prop}\label{prop:Reeb-vector-fields}
    On every cocontact manifold $(M, \tau, \eta)$ there exist two distinguished vector fields $\Rt$, $\Rz$ on $M$ such that
    \begin{gather*}
        \contr{\Rt}\tau = 1\,,\qquad \contr{\Rt}\eta = 0\,,\qquad \contr{\Rt}\d\eta = 0\,,\\
        \contr{\Rz}\tau = 0\,,\qquad \contr{\Rz}\eta = 1\,,\qquad \contr{\Rz}\d\eta = 0\,,
    \end{gather*}
    or, equivalently, $\Rt = \flat^{-1}(\tau)$ and $\Rz = \flat^{-1}(\eta)$. These vector fields $\Rt$ and $\Rz$ are called \emph{time and contact Reeb vector fields} respectively.
\end{prop}

\begin{thm}[Cocontact Darboux theorem]\label{thm:Darboux-cocontact}
	Given a cocontact manifold $(M,\tau,\eta)$, around every point $x\in M$ there exist local coordinates $(t, q^i, p_i, z)$ such that
	$$ \tau = \d t\ ,\quad \eta = \d z - p_i\d q^i\,. $$
	These coordinates are called \emph{canonical} or \emph{Darboux} coordinates. In addition, in Darboux coordinates, the Reeb vector fields read
	$$ \Rt = \parder{}{t}\ ,\quad \Rz = \parder{}{z}\,. $$
\end{thm}

\begin{prop}\label{prop_Jacobi_mfold}
    Let $(M,\tau,\eta)$ be a cocontact manifold. Then, $(M,\Lambda,E)$ is a Jacobi manifold, where
    $$ \Lambda(\alpha,\beta) = -\d\eta(\sharp\alpha,\sharp\beta)\ ,\quad E = -\Rz\,. $$
\end{prop}

The bivector $\Lambda$ induces a $\Cinfty(M)$-module morphism $\hat\Lambda\colon\Omega^1(M)\to\X(M)$ given by
\begin{equation}
    \hat\Lambda(\alpha) = \Lambda(\alpha,\cdot) = \sharp\alpha - \alpha(\Rz)\Rz - \alpha(\Rt)\Rt\,. 
    \label{Lambda_hat_cocontact}
\end{equation}
It can be seen that $\ker\hat\Lambda = \langle\tau,\eta\rangle$. The morphism $\hat \Lambda$ is also denoted by $\sharp_\Lambda$ in the literature \cite{deLeon2019a,deLeon2021a}.

Taking Darboux coordinates $(t,q^i,p_i,z)$, the bivector $\Lambda$ has local expression
$$ \Lambda = \parder{}{q^i}\wedge\parder{}{p_i} - p_i\parder{}{p_i}\wedge\parder{}{z}\,, $$
and the Jacobi bracket reads
$$ \{f,g\} = \parder{f}{q^i}\parder{g}{p_i} - \parder{g}{q^i}\parder{f}{p_i} - \left( \parder{f}{p_i}\parder{g}{z} - \parder{g}{p_i}\parder{f}{z} \right) - f\parder{g}{z} + g\parder{f}{z}\,. $$
In particular, one has
$$ \{q^i,q^j\} = \{p_i, p_j\} = 0\,,\qquad \{q^i, p_j\} = \delta_j^i\,,\qquad \{q^i, z\} = -q^i\,,\qquad \{p_i,z\} = -2p_i\,. $$

\subsection{Cocontact Hamiltonian systems}

\begin{dfn}
	A \emph{cocontact Hamiltonian system} is family $(M,\tau,\eta,H)$ where $(\tau,\eta)$ is a cocontact structure on $M$ and $H:M\to\R$ is a Hamiltonian function. The \emph{cocontact Hamilton equations} for a curve $\psi\colon I\subset \R\to M$ are
	\begin{equation}\label{eq:Ham-eq-cocontact-sections}
		\begin{dcases}
			\contr{\psi'}\d \eta = \Big(\d H-(\Lie_{R_s}H)\eta-(\Lie_{\Rt}H)\tau\Big)\circ\psi\,,
			\\
			\contr{\psi'}\eta = -H\circ\psi\,,
			\\
			\contr{\psi'}\tau = 1\,,
		\end{dcases}
	\end{equation}
	where $\psi'\colon I\subset\R\to\T M$ is the canonical lift of $\psi$ to the tangent bundle $\T M$. The \emph{cocontact Hamiltonian equations} for a vector field $X_H\in\X(M)$ are:
	\begin{equation}\label{eq:Ham-eq-cocontact-vectorfields}
		\begin{dcases}
			   \contr{X_H}\d \eta = \d H-(\Lie_{R_s}H)\eta-(\Lie_{\Rt}H)\tau\,,
			\\
			\contr{X_H}\eta = -H\,,
			\\
			\contr{X_H}\tau = 1\,,
		\end{dcases}
	\end{equation}
	or equivalently, $\flat(X_H)=\d H-\left(\Lie_{R_s}H+H\right)\eta+\left(1-\Lie_{\Rt}H\right)\tau$. The unique solution to these equations is called the \emph{cocontact Hamiltonian vector field}.
\end{dfn}

Given a curve $\psi$ with local expression $\psi(r)=(f(r),q^i(r),p_i(r),z(r))$, the third equation in \eqref{eq:Ham-eq-cocontact-sections} imposes that $f(r)=r + c$, where $c$ is some constant, thus we will denote $r\equiv t$, while the other equations read:
\begin{equation}\label{eq:Hamilton-cocontact}
   \begin{dcases}
		\dot q^i =\frac{\partial H}{\partial p_i}\,,
		\\
	   \dot p_i = -\left(\frac{\partial H}{\partial q^i}+p_i\frac{\partial H}{\partial z}\right)\,,
		\\
	   \dot z = p_ i\frac{\partial H}{\partial p_i}-H\,.		
	\end{dcases}	  
\end{equation}
On the other hand, the local expression of the cocontact Hamiltonian vector field is
$$ X_H = \parder{}{t} + \parder{H}{p_i}\parder{}{q^i} - \left(\parder{H}{q^i} + p_i\parder{H}{z}\right)\parder{}{p_i} + \left(p_i\parder{H}{p_i} - H\right)\parder{}{z}\,. $$

The integral curves of the cocontact Hamiltonian vector field satisfy the following variational principle~\cite{Liu2018}, which is a Hamiltonian version of the Herglotz principle~\cite{Herglotz1930}.
\begin{thm}[Hamiltonian formulation of the Herglotz principle]
Given a Hamiltonian $H: \R \times \cT Q \times \R \to \R$, a curve $c=(\Id_\R,q,p,z):[0,T] \to \cT Q \times \R$ is an integral curve of the Hamiltonian vector field $X_H$ if and only if it is a critical point of the action map
\begin{equation}
    \mathcal{A}(c) = \int_0^T \big(p(t) \dot{q}(t) - H(t,q(t),p(t),z(t)) \big)\; \dd t
    \label{action_Hamiltonian}
\end{equation}
among all curves satisfying $c(0)= c_0$, $c(T) = c_T$ and $\dot{z} = p(t) \dot{q}(t) - H(t,q(t),p(t),z(t))$.
\end{thm}

\section{Symmetries and dissipated quantities in cocontact systems} \label{section_symmetries}
There are several notions of symmetries in contact mechanics depending on the structures they preserve \cite{deLeon2020a,Gaset2020a}. However, in the present paper we will restrict ourselves to what we call generalized dynamical symmetries (see \cite{Gaset2022} for other notions of symmetry). In some cases we will restrict ourselves to the case of cocontact manifolds of the form $M = \R\times N$ where $N$ is a contact manifold (see Example \ref{exmpl:R-contact}). In this case, the natural projection $\R\times N\to\R$ defines a global canonical coordinate $t$.

\begin{dfn}
    Let $(M,\tau,\eta,H)$ be a cocontact Hamiltonian system and let $X_H$ be its cocontact Hamiltonian vector field.
    \begin{itemize}
        \item If $M = \R\times N$ with $N$ a contact manifold, a \textbf{generalized dynamical symmetry} is a diffeomorphism $\Phi\colon M\to M$ such that $\eta(\Phi_\ast X_H) = \eta(X_H)$ and $\Phi^\ast t = t$.
        \item An \textbf{infinitesimal generalized dynamical symmetry} is a vector field $Y\in\X(M)$ such that $\eta([Y,X_H]) = 0$ and $\contr{Y}\tau = 0$. In particular, if $M = \R\times N$ with $N$ a contact manifold, the flow of $Y$ is made of generalized dynamical symmetries.
    \end{itemize}
\end{dfn}


\begin{dfn}
    Let $(M,\tau,\eta,H)$ be a cocontact Hamiltonian system. A \emph{dissipated quantity} is a function $f\in\Cinfty(M)$ such that
    $$ X_H f = -(\Rz H) f\,. $$
\end{dfn}

It is worth pointing out that, unlike in the contact case, the Hamiltonian function is not, in general, a dissipated quantity. Indeed, using that
$$ X_H H = -(\Rz H)H + \Rt H\,, $$
it is clear that $H$ is a dissipated quantity if and only if it is time-independent, i.e. $R_t H = 0$. This resembles the cosymplectic case, where the Hamiltonian function is conserved if, and only if, it is time-independent (see \cite{Cantrijn1992}).

\begin{prop}
    A function $f\in \Cinfty(M)$ is a dissipated quantity if and only if $\{f, H\} = \Rt(f)$, where $\{\cdot,\cdot\}$ is the Jacobi bracket associated to the cocontact structure $(\tau,\eta)$.
\end{prop}

\begin{proof}
    The Jacobi bracket of $f$ and $H$ is given by
    \begin{equation}
    \begin{aligned}
        \left\{f,H\right\} = \Lambda(\d f, \d H) + f E(H) - H E(f)
        = -\d \eta \left(\sharp \d f, \sharp \d H\right) - f \Rz(H) + H \Rz(f)\,,
    \end{aligned}
    \end{equation}
    but
    \begin{equation}
        \sharp \d f = X_f + \left(\Rz(f) + f\right) \Rz - \left(1 - \Rt(f) \right) \Rt\,,
    \end{equation}
    so
    \begin{equation}
        \contr{\sharp\d f} \d \eta = \contr{X_f} \d \eta
        = \d f - \Rz(f) \eta -\Rt(f) \tau\,,
    \end{equation}
    and thus
    \begin{equation}
        \d \eta (\sharp \d f, \sharp \d H) 
        = X_H(f)  + \Rs(f) H - \Rt(f)
        \,.
    \end{equation}
    Hence,
    \begin{equation}
         \left\{H,f\right\} + \Rt(f) = X_H(f)  + \Rs(H) f
         \, .
    \end{equation}
    In particular, the right-hand side vanishes if and only if $f$ is a dissipated quantity.
\end{proof}


\begin{thm}[Noether's theorem]\label{Noether_thm}
    Let $Y$ be an infinitesimal generalized dynamical symmetry of the cocontact Hamiltonian system $(M,\tau,\eta,H)$. Then, $f = -\contr Y\eta$ is a dissipated quantity of the system. Conversely, given a dissipated quantity $f\in\Cinfty(M)$, the vector field $Y = X_f - R_t$, where $X_f$ is the Hamiltonian vector field associated to $f$, is an infinitesimal generalized dynamical symmetry and $f = -\contr Y\eta$.
\end{thm}

\begin{proof}
    Let $f = -\contr{Y}\eta$, where $Y$ is an infinitesimal generalized dynamical symmetry. Then,
    \begin{align*}
        \Lie_{X_H}f &= -\Lie_{X_H}\contr{Y}\eta
        = -\contr{Y}\Lie_{X_H}\eta - \contr{[X_H,Y]}\eta
        =\\
        &\contr{Y}\left( \Rz(H)\eta + \Rt(H)\tau \right)
        = \Rz(H)\contr{Y}\eta
        = -\Rz(H) f\,,
    \end{align*}
    and thus $f$ is a dissipated quantity.

    On the other hand, given a dissipated quantity $f$, let $Y = X_f - \Rt$. Then, it is clear that $f = -\contr{Y}\eta$. In addition, $\contr{Y}\tau = 0$, and
    \begin{align*}
        \contr{[X_H,Y]}\eta 
        &= \liedv{X_H} \contr{Y} \eta -  \contr{Y} \liedv{X_H} \eta
        = -\liedv{X_H} f  + \contr{Y} \left(\Rz(H) \eta +\Rt(H) \tau\right)
        \\
        &= \Rz(H) f - \Rz(H) \contr{Y} \eta = 0\, .
    \end{align*}
\end{proof}


The symmetries presented yield dissipated quantities. However, we are also interested in finding conserved quantities. The latter are important due to their elation with complete solutions of the Hamilton--Jacobi problem (see Section \ref{section_HJ_action_dep}).

\begin{dfn}
    A \emph{conserved quantity} of a cocontact Hamiltonian system $(M,\tau,\eta,H)$ is a function $g\in\Cinfty(M)$ such that
    $$ X_H g = 0\,. $$
\end{dfn}

Taking into account that every dissipated quantity changes with the same rate $\Rz(H)$, we have the following result, whose proof is straightforward.

\begin{prop}\label{prop:conserved-dissipated}
    Consider a cocontact Hamiltonian system $(M,\tau,\eta,H)$.
    \begin{itemize}
        \item If $f_1,f_2$ are two dissipated quantities and $f_2\neq 0$, then $f_1/f_2$ is a conserved quantity.
        \item If $f$ is a dissipated quantity and $g$ is a conserved quantity, then $fg$ is a dissipated quantity.
    \end{itemize}
\end{prop}

\section{The action-independent approach}\label{section_HJ_action_indep}

\subsection{Hamilton--Jacobi theory. The action-independent approach}

 
  Let $(\R\times \T^*Q\times \R, \tau, \eta, H)$ be a cocontact Hamiltonian system, where $\tau = \d t$, $\eta = \d z - {\theta_0}$ and $\theta_0 = p_i\d q^i$ is the Liouville one-form of the cotangent bundle. Consider a section $\gamma$ of the bundle $\pi_Q^t: \R\times \T^*Q\times \R\to \R\times Q$, locally given by
  \begin{equation}
    \begin{aligned}
        \gamma: \R\times Q &\longrightarrow  \R\times \T^*Q\times \R\\
        \left(t, q^i\right) &\longmapsto \left(t, q^i, \gamma_i(t,q), S(t,q)\right) \,.
    \end{aligned}
  \end{equation}
  Let us introduce the vector field $X_H^\gamma$ on $\R\times Q$ given by
  \begin{equation}
      X_H^\gamma = \T \pi_Q^t \circ X_H \circ \gamma\, ,
  \end{equation}
  where $X_H$ is the Hamiltonian vector field of $(\R\times \T^*Q\times \R, \tau, \eta, H)$. Suppose that $X_H^\gamma$ and $X_H$ are $\gamma$-related, i.e.,
  \begin{equation}
      X_H \circ \gamma = \T \gamma \circ X_H^\gamma\, ,
      \label{eq_gamma_related}
  \end{equation}
  so that the following diagram commutes:
\begin{center}
    \begin{tikzcd}
        \R \times \cT Q \times \R \arrow[d, "\pi_Q^t"] \arrow[r, "X_H"]
        & \T(\R \times \cT Q\times \R) \arrow[d, swap, "\T\pi_Q^t"] \\
         \R \times Q \arrow[r, "X_H^\gamma"] \arrow[u, "\gamma", bend left]     & \T (\R \times Q) \arrow[u, swap, "\T \gamma", bend right]
    \end{tikzcd}
\end{center}
Locally,
\begin{equation}
    X_H \circ \gamma = \frac{\partial}{\partial t} 
    + \frac{\partial H}{\partial p_i} \frac{\partial }{\partial q^i} 
    - \left( \frac{\partial H }{\partial q^i} + \gamma_i \frac{\partial H}{\partial z} \right) \frac{\partial }{\partial p_i}
    + \left(\gamma_i \frac{\partial H}{\partial p_i} -H\right) \frac{\partial }{\partial z}\,,
\end{equation}
and
\begin{equation}
   \T\gamma \circ X_H^\gamma = \frac{\partial}{\partial t} 
    + \frac{\partial H}{\partial p_i} \frac{\partial }{\partial q^i}
    +\left(\frac{\partial \gamma_i} {\partial t} + \frac{\partial H} {\partial p_j} \frac{\partial \gamma_j} {\partial q^i} \right) \frac{\partial  } {\partial p_i}
    + \left(\frac{\partial S}{\partial t} + \frac{\partial S}{\partial q^i}\frac{\partial H}{\partial p_i} \right) \frac{\partial }{\partial z}\, ,
\end{equation}
so equation~\eqref{eq_gamma_related} holds if and only if
\begin{equation}\label{Sardonian_generalized_HJ_local}
    \begin{aligned} 
       - \left( \frac{\partial H }{\partial q^i} + \gamma_i \frac{\partial H}{\partial z} \right)
       &= \frac{\partial \gamma_i} {\partial t} + \frac{\partial H} {\partial p_j} \frac{\partial \gamma_j} {\partial q^i} \, ,\\
       \gamma_i \frac{\partial H}{\partial p_i} - H
       &= \frac{\partial S}{\partial t} + \frac{\partial S}{\partial q^i}\frac{\partial H}{\partial p_i}\, .
    \end{aligned}
\end{equation}
    


\begin{dfn}
   Given a section $\alpha: \R \times Q \to \R \times \bigwedge\nolimits^k\cT Q$ and $t \in \mathbb{R}$, let 
   \begin{equation}
       \begin{aligned}
           \alpha_{(t)} :  Q &\longrightarrow \bigwedge\nolimits^k\cT Q\\
           x &\longmapsto \pr_{\Lambda^k \cT Q} (\alpha(t,x))\, ,
       \end{aligned}
   \end{equation}
   where $\pr_{\bigwedge\nolimits^k \cT Q}: \R \times \bigwedge\nolimits^k\cT Q \to \bigwedge\nolimits^k\cT Q$ is the canonical projection. The \emph{exterior derivative of $\alpha$ at fixed $t$} is the section of $\R \times \bigwedge\nolimits^{k+1}\cT Q \to \R \times Q$ given by
   \begin{equation}
       \d_Q \alpha(t, x) = (t, \d \alpha_{(t)} (x))\,.
   \end{equation}
\end{dfn}

In coordinates, for $f\in\Cinfty(\R\times Q)$ and $\alpha(t,x) = (t,\alpha_i\d_x q^i)$ a section of the bundle $\R \times Q\to \R \times \bigwedge\nolimits^k\cT Q$, the local expressions are
\begin{equation}
    \begin{split}
        \d_Q f &= \left(t, \frac{\partial f}{\partial  q^i} \d_x q^i\right)\,,\\
        \d_Q \alpha &= \left(t, \frac{\partial \alpha_j}{\partial q^i} \d_x q^i \wedge \d_x q^j \right)\,.
    \end{split}
\end{equation}
Since we shall be considering fixed $t$, we will often make the abuse of notation $$ \d_Q f = \dfrac{\partial f}{\partial q^i} \d_x q^i\,. $$

\begin{dfn}
     Given $f\in \Cinfty(\R\times Q)$, the \emph{$1$-jet of $f$ at fixed $t$} is the section $j_t^1 f\colon \R \times Q \to \R \times \cT Q \times \R$ given by
      $$ j_t^1 f(t, x) =(\d_Q f, f)\,. $$
\end{dfn}

Let us recall that a \emph{Legendrian submanifold} $N\hookrightarrow M$ of a $(2n+1)$-dimensional contact manifold $(M, \eta)$ is an $n$-dimensional submanifold such that $\restr{\eta}{N} = 0$ (see \cite{deLeon2021a}).

\begin{prop}\label{proposition_Legendrian_jet}
     Let $\gamma$ be a section of $\pi_Q^t: \mathbb{R} \times \cT Q \times \R \to \R\times Q$. Then, for every $t \in \mathbb{R}$, $\Ima \gamma(t, \cdot)$ is a Legendrian submanifold of $(\cT Q\times \R, \eta)$ if and only if it is the image of the 1-jet at fixed $t$ of a function, namely, 
     $$ \gamma(t,x) = j_t^1 f(t, x) = (\d_Q f, f) \,.$$
\end{prop}
\begin{proof}
    Let $t \in \mathbb{R}$ and let $\gamma: \R\times Q \to \R \times \cT Q \times \R$ such that $\gamma(t,q) = (t, \alpha(t,q),f(t,q))$. Clearly, $\gamma^*\tau = 0$, hence $\Ima \gamma$ is Legendrian if and only if $\gamma^* \eta = 0$. Thus,
    \begin{equation}
        \gamma^* \eta = f^*\dd z - \alpha^* \theta_Q = \dd_Q f - \alpha\,,
    \end{equation}
    so $\gamma^* \eta$ vanishes precisely when $\alpha= \dd_Q f$.
\end{proof}

Now, suppose that $\Ima \gamma$ is a Legendrian submanifold. By Proposition \ref{proposition_Legendrian_jet}, we have that
\begin{equation}
    \gamma_i = \frac{\partial S}{\partial q^i}\,,
\end{equation}
so equations~\eqref{Sardonian_generalized_HJ_local} can be written as
\begin{subequations}
\begin{align}
    - \left( \frac{\partial H }{\partial q^i} + \frac{\partial S}{\partial q^i} \frac{\partial H}{\partial z} \right)
   &= \frac{\partial^2 S}{\partial t \partial q^i} + \frac{\partial H} {\partial p_j} \frac{\partial S}{\partial q^i \partial q^j} \, ,
   \label{Sardonian_HJ_local_a}
   \\
    \frac{\partial S}{\partial q^i} \frac{\partial H}{\partial p_i} - H
   &= \frac{\partial S}{\partial t} + \frac{\partial S}{\partial q^i}\frac{\partial H}{\partial p_i}\, .
   \label{Sardonian_HJ_local_b}
\end{align}    
\label{Sardonian_HJ_local}
\end{subequations}
equation~\eqref{Sardonian_HJ_local_a} implies that
\begin{equation}
    \d_Q(H\circ j^1_t S)  + \d_Q \left(\Rt S \right) = 0\,,
    \label{eq:HJ_Sardonian_redundant}
\end{equation}
while equation~\eqref{Sardonian_HJ_local_b} yields
\begin{equation}
    H = - \frac{\partial S}{\partial t}\, ,
\end{equation}
that is,
\begin{equation}
    H \circ j^1_t S + \frac{\partial S}{\partial t} = 0\,.
    \label{eq:HJ_Sardonian}
\end{equation}
Clearly, equation~\eqref{eq:HJ_Sardonian_redundant} is implied by equation~\eqref{eq:HJ_Sardonian}. We have thus proven the following.

\begin{thm}[Action-independent Hamilton--Jacobi theorem]\label{thm_action_indep}
Let $\gamma$ be a section of $\pi_Q^t: \R\times \T^*Q\times \R\to \R\times Q$ such that, for every $t\in \R$, $\Ima \gamma(t, \cdot)$ is a Legendrian submanifold of $(\cT Q\times \R, \eta)$. Then, $X_H^\gamma$ and $X_H$ are $\gamma$-related if and only if equation~\eqref{eq:HJ_Sardonian} holds. This equation will be called the \emph{action-independent Hamilton--Jacobi} equation for $(\R\times \T^*Q\times \R, \tau, \eta, H)$. The function $S$ such that $\gamma = j_t^1 S$ is called a \emph{generating function} for $H$.
\end{thm}

In order to study the integrability of cocontact Hamiltonian systems, it is of interest to introduce the following.

\begin{dfn}\label{def:complete_sols_indep}
Let $(\R\times \T^*Q\times \R, \tau, \eta, H)$ be a cocontact Hamiltonian system.
A \emph{complete solution of the action-independent Hamilton--Jacobi problem} for $(\R\times \T^*Q\times \R, \tau, \eta, H)$ is a local diffeomorphism $\Phi\colon \R\times Q \times \R^{n+1} \to \R \times \cT Q \times \R$ such that, for each $\lambda \in \R^{n+1}$,
\begin{equation}
\begin{aligned}
    \Phi_\lambda \colon \R\times Q &\to \R \times\cT Q \times \R  \\
    \left(t, q^i\right) &\mapsto \Phi\left(t, q^i, \lambda\right) 
\end{aligned}    
\end{equation}
is a solution of the action-independent Hamilton--Jacobi problem for $(\R\times \cT Q\times \R, \tau, \eta, H)$.
\end{dfn}

It is worth noting that complete solutions depend on $n+1$ real parameters, one extra parameter in comparison with the (co)symplectic case. In order to consider complete solutions depending on just $n$ parameters, we shall introduce a different approach to the Hamilton--Jacobi problem for (co)contact Hamiltonian systems (see Section \ref{section_HJ_action_dep}).

Let $\alpha\colon \R\times Q\times \R^{n+1}\rightarrow \mathbb{R}^{n+1}$, and $\pi_i\colon \R^{n+1} \to \R$ denote the canonical projections. One can define the $n+1$ functions $f_i=\pi_i\circ \alpha\circ \Phi^{-1}$ on $\R \times \cT Q \times \R$, so that the following diagram commutes:
\[
    \begin{tikzcd}
        \R \times Q \times \R^{n+1} \arrow[r, "\Phi", shift left] \arrow[d, "\alpha"] & \R\times \cT Q \times \R \arrow[l, "\Phi^{-1}", shift left] \arrow[d, "f_i"] \\
        \R^{n+1} \arrow[r, "\pi_i"] & \R
    \end{tikzcd}
\]

\begin{thm}\label{thm:complete_solution_involution_indep}
   Let $\Phi\colon \R\times Q \times \R^{n+1} \to \R \times \cT Q \times \R$ be a complete solution of the action-independent Hamilton--Jacobi problem for $(\R\times \T^*Q\times \R, \tau, \eta, H)$. Then, 
    \begin{enumerate}[{\rm (i)}]
        \item For each $i\in \{1, \dotsc, n+1\}$, the function $f_i=\pi_i\circ \alpha\circ \Phi^{-1}$ is a constant of the motion. However, these functions are not necessarily in involution, i.e., $\{f_i, f_j\}\neq 0$.
        \item For each $i\in \{1, \dotsc, n+1\}$, the function $\hat f_i= g f_i$, where $g$ is a dissipated quantity, is also a dissipated quantity. Moreover, if $\Rt(H) = 0$ and taking $g=H$, these functions are in involution, i.e., $\{\hat f_i, \hat f_j\}= 0$.
    \end{enumerate}
\end{thm}

\begin{proof}
We can write 
$$
\Ima \Phi_\lambda = \{ x \in \R\times \cT Q \times \R \mid f_i(x) = \lambda_i,\  i=1, \dotsc, n\} = \bigcap_{i=1}^n \, f_i^{-1}(\lambda_i)\,,
$$
where $\lambda = (\lambda, \dotsc, \lambda_n) \in \R^n$.
Since $X_H$ is tangent to any of the submanifolds $\Ima \Phi_\lambda$, we deduce that
$$
X_H f_i = 0\,,
$$
so each of the functions $f_i$, for $i=1, \ldots, n$, is a constant of the motion.

On the other hand, we can compute
$$
\{f_i, f_j\} = X_{f_j} (f_i) - \Rt (f_i) -f_i \Rz(f_j)
\,,
$$
which does not vanish in general.
By Proposition \ref{prop:conserved-dissipated}, the product of a conserved quantity and a dissipated quantity is a dissipated quantity. Let $f_i$ and $f_j$ be conserved quantities and take $g = H$. Then,
\begin{equation}
\begin{aligned}
    \{\hat f_i, \hat f_j \}
    &= \{H f_i, H f_j \} = f_j \{ Hf_i, H \} + H\{Hf_i,f_j\} - f_jH\Rz(Hf_i) \\
    &= - f_jH\{H,f_i\} + f_if_jH\Rz(H) - f_iH\{f_j,H\} - H^2\{f_j,f_i\} + f_iH^2\Rz(f_j) - f_jH\Rz(Hf_i)\\
    &= 0\,.
\end{aligned}
\label{eq_involution_dissipated}
\end{equation}

\end{proof}

\begin{rmrk}\label{remark_LA}
Making use of a symplectization, one can study a time-independent contact Hamiltonian system as an exact symplectic Hamiltonian system with one additional dimension. Dissipated quantities in involution with respect to the Jacobi bracket of the contact system lead to conserved quantities in involution with respect to the Poisson bracket of the associated symplectic system. On the other hand, the celebrated Liouville--Arnold Theorem \cite{Arnold1978} permits to construct action-angle coordinates of a $2n$-dimensional symplectic Hamiltonian system with $n$ conserved quantities in involution, leading to integrability by quadratures. Therefore, dissipated quantities in involution could lead to integrability by quadratures of contact Hamiltonian systems. However, this relation is highly non-trivial and it is subject of further research. An alternative approach to Hamilton--Jacobi theory and integrability by quadratures for contact Hamiltonian systems can be found in \cite{Grillo2020}.
\end{rmrk}


{
Complete solutions of the Hamilton--Jacobi problem {may} be used to integrate the dynamics of the system as follows:
\begin{enumerate}[{\rm (i)}]
    \item Solve the Hamilton--Jacobi equation
    \begin{equation}
        H \circ j^1_t S_\lambda + \parder{S_\lambda}{t} = 0
    \end{equation}
    for arbitrary values of $\lambda \in \R^{n+1}$. Let $\Phi_\lambda = j_t^1 S_\lambda$.
    \item Compute the integral curves $\sigma\colon \R \to \R \times Q,\ \sigma(t) = (t, q^i(t))$ of $X_H^\gamma$, which are given by
    \begin{equation}
        \frac{\d q^i}{\d t} = \restr{\frac{\partial H}{\partial p_i}}{\Ima \Phi_\lambda}\, , \label{Eq_integrate}
    \end{equation}
    where the restriction to $\Ima \Phi_\lambda$ means that one has to write $p_i=\tparder{S_\lambda}{q^i}$ and $z=S_\lambda$.
    \item The integral curves $\tilde \sigma$ of $X_H$ on $\Im \Phi_\lambda$ are given by $\Phi_\lambda \circ \sigma$, namely,
    \begin{equation}
        \tilde \sigma(t) =\Phi_\lambda \circ \sigma(t)= \left(\sigma(t), \parder{S_\lambda}{q^i} (\sigma(t)), S_\lambda (\sigma(t)) \right)\, .
    \end{equation}
\end{enumerate}
}

It is worth noting that computing the integral curves of $X_H^\gamma$ is not always straightforward. 
However, there are some relevant cases in which it is particularly easy.

\begin{exmpl}
Suppose that $Q=\R^n$ with the Euclidean norm.
If the generating function is \emph{separable}, i.e., $S(t, q^1, q^2, \ldots,q^n) = S_0(t) + S_1(q^1)+\dotsb + S_n(q^n)$, and the Hamiltonian is mechanical, namely, $H=\dfrac{\norm{p}^2}{2m(t)} + V(t,q,z)$, then equations~\eqref{Eq_integrate} simplify to
\begin{equation}
       \frac{\d q^i}{\d t} = \frac{1}{m(t)} S_i'(q^i)\, .
\end{equation}
\end{exmpl}

\subsection{Example: the free particle with time-dependent mass and a linear external force}
    Consider the cocontact Hamiltonian system $(\R\times \cT Q \times \R, \d t, \eta, H)$, where
    \begin{equation}
        H = \frac{p^2}{2m(t)} - \frac{\kappa}{m(t)} z,
    \end{equation}
    with $m$ a function depending only on $t$, expressing the mass of the particle, and $\kappa$ a positive constant. 
    Then, the action-independent Hamilton--Jacobi equation for $H$ is given by 
    \begin{equation}
        \frac{1}{2m(t)}\left( \frac{\partial S}{\partial q}\right)^2 - \frac{\kappa}{m(t)} S(t,q) + \frac{\partial S}{\partial t} = 0\, ,
    \end{equation}
    that is,
     \begin{equation}
       \left( \frac{\partial S}{\partial q}\right)^2 - 2 \kappa S(t,q) + 2m(t)\frac{\partial S}{\partial t} = 0\, .
        \label{HJ_example_indep}
    \end{equation}
    Suppose that the generating function $S$ is separable, namely, $S(t,q)= \alpha(t)+\beta(q)$. Then, equation~\eqref{HJ_example_indep} can be written as
    \begin{equation}
         \left( \frac{\d \beta}{\d q}\right)^2 - 2 \kappa \alpha(t) - 2 \kappa \beta(q) + 2m(t)\frac{\d \alpha}{\d t} = 0\, ,
    \end{equation}
    so
   \begin{align}
       &  2m(t)\frac{\d \alpha}{\d t} -2 \gamma \alpha(t) = 0\,,\\
       & \left( \frac{\d \beta}{\d q}\right)^2 - 2 \kappa \beta(q) = 0.
   \end{align}
   Then,
   \begin{equation}
       \alpha_{\lambda_1}(t) = \lambda_1 e^{\kappa \int_0^t \frac{1}{m(s)} \d s}\,\qquad
       \beta_{\lambda_2}(q) = \left( \sqrt{\frac{\kappa}{2}} q + \lambda_2 \right )^2,
   \end{equation}
   that is,
   \begin{equation}
       S_\lambda(t,q) = \lambda_1 e^{\kappa \int_0^t \frac{1}{m(s)} \d s} + \left( \sqrt{\frac{\kappa}{2}} q + \lambda_2 \right )^2,
   \end{equation}
   and 
   $$\Phi(t,q,\lambda) = j_t^1 S_\lambda(t,q) = \left(t, q, \sqrt{2\kappa}\left( \sqrt{\frac{\kappa}{2}} q + \lambda_2 \right),  \lambda_1 e^{\kappa \int_0^t \frac{1}{m(s)} \d s} + \left( \sqrt{\frac{\kappa}{2}} q + \lambda_2 \right )^2\right)$$ 
   is a complete solution. Its inverse is given by
   \begin{equation}
       \Phi^{-1} \colon (t,q, p, z) \mapsto
       \left(t, q, e^{-\kappa \int_0^t \frac{1}{m(s)} \d s} \left(z - \frac{p^2}{2\kappa}  \right),\frac{p-\kappa q}{\sqrt{2\kappa}}\right)\, .
   \end{equation}
   Hence,
   \begin{equation}
       f_1 (t, q, p,z) = e^{-\kappa \int_0^t \frac{1}{m(s)} \d s} \left(z - \frac{p^2}{2\kappa} \right)\, ,
   \end{equation}
   and
   \begin{equation}
       f_2 (t, q, p,z) = \frac{p-\kappa q}{\sqrt{2\kappa}}
   \end{equation}
   are conserved quantities.
   
   {The Hamiltonian vector field of $H$ is given by
   \begin{equation}
       X_H 
     = \parder{}{t} + \frac{p}{m(t)}\parder{}{q} + \frac{\kappa p}{m(t)} \parder{}{p} + \left(\frac{p^2}{2m(t)} + \frac{\kappa}{m(t)} z\right)\parder{}{z}\,.
   \end{equation}
   One can check that $X_H(f_1)=X_H(f_2)=0$. Moreover,
    \begin{equation}
       X_H^\gamma
     = \restr{\parder{}{t} + \frac{p}{m(t)}\parder{}{q}}{\Ima \Phi_\lambda}
     = \parder{}{t} + \frac{\sqrt{2\kappa}\left( \sqrt{\frac{\kappa}{2}} q + \lambda_2 \right)}{m(t)}\parder{}{q}
     \,,
   \end{equation}
   whose integral curves $\sigma(t)=(t, q(t))$ are given by
   \begin{equation}
       q(t) = e^{\int _1^t\frac{\kappa }{m(s)}\d s} \left(\int _1^t\frac{\sqrt{2\kappa} e^{-\int _1^u\frac{\kappa }{m(s)}\d s }  \lambda }{m(u)}\d u+c\right)\, ,
   \end{equation}
   where $c$ is a constant.
   Then, the integral curves of $X_H$ along $\Ima \Phi_\lambda$ are given by $\Phi_\lambda \circ \sigma(t) = (t, q(t), p(t), z(t))$, where
   \begin{equation}
       p(t) = \sqrt{2\kappa}\left( \sqrt{\frac{\kappa}{2}} q(t) + \lambda_2 \right)\, ,
   \end{equation}
   and
   \begin{equation}
       z(t) = \lambda_1 e^{\kappa \int_0^t \frac{1}{m(s)} \d s} + \left( \sqrt{\frac{\kappa}{2}} q(t) + \lambda_2 \right)^2\, .
   \end{equation}
   }

\subsection{The variational interpretation of the solution to Hamilton--Jacobi equation}
Suppose that $\sigma\colon [0,T]\to Q$ is a trajectory given by the cocontact Hamilton equations \eqref{eq:Hamilton-cocontact} for the Hamiltonian function $H\colon \R \times \cT Q\times \R \to \R $.
If $\gamma= j_t^1 S$ is a solution to the Hamilton--Jacobi problem for $H$, the generating function $S$ can be interpreted as the action of the lifted curve $j_t^1 S \circ \sigma$ up to a constant.
\begin{thm}
Suppose that $S\in \Cinfty(\R \times Q)$ is a generating function for $H$.
Let $\sigma\colon[0,T]\to Q$ be a curve with local expression $\sigma(t) = (q^i(t))$ such that $c = (\mathrm{Id},\sigma)\colon t\in\R \mapsto (t,\sigma(t))\in\R\times Q$ is an integral curve of $X_H^\gamma$. Then,
\begin{equation}
    (S\circ c)(t) = \mathcal{A}(j_t^1 S \circ \sigma)(t) + S_0\,,
\end{equation}
for some $ S_0 \in \R$, where $\mathcal{A}$ denotes the action map \eqref{action_Hamiltonian}.
\end{thm}
\begin{proof}
Assume that $S\in \Cinfty(\R \times Q)$ is a generating function for $H$. Then,
 \begin{equation}
 \begin{aligned}
          \frac{\dd}{\dd t}  S(t,q(t)) &=  \frac{\partial  S(t,\sigma(t))}{ \partial t} + \frac{\partial  S(t,\sigma(t))}{ \partial q^i} \dot{q}^i(t) \\ &=
            \frac{\partial  S(t,\sigma(t))}{ \partial q^i} \dot{q}^i(t) - H \circ j_t^1 S \circ \sigma(t) \\&=
           \frac{\dd}{\dd t}  \mathcal{A}(j_t^1 S \circ \sigma)(t)\,,
 \end{aligned}
 \end{equation}
 where we have used the Hamilton--Jacobi equation~\eqref{eq:HJ_Sardonian} on the second step, and the definition of the action map \eqref{action_Hamiltonian} on the last step. Hence,
 \begin{equation*}
    S(t,q(t)) = \mathcal{A}(j_t^1 S \circ q)(t) + S_0\,,
 \end{equation*}
for some constant $S_0$.
\end{proof}

\subsection{A new approach for the Hamilton--Jacobi problem in time-independent contact Hamiltonian systems}

Let us recall that a contact Hamiltonian system $(M, \eta, H)$ is a contact manifold $(M, \eta)$ together with a Hamiltonian function $H\colon M \to \R$ (see \cite{Gaset2020a,deLeon2019a}).
The Hamiltonian vector field of $H$ is locally given by
$$ X_H =  \parder{H}{p_i}\parder{}{q^i} - \left(\parder{H}{q^i} + p_i\parder{H}{z}\right)\parder{}{p_i} + \left(p_i\parder{H}{p_i} - H\right)\parder{}{z}\,. $$

The analogous of Theorem \ref{thm_action_indep} for autonomous contact Hamiltonian systems was developed in \cite{deLeon2017} (see also \cite{deLeon2021d}): 

\begin{thm}[Hamilton--Jacobi Theorem for autonomous systems]\label{thm_action_indep_autonomous}
Let $(\cT Q\times \R, \eta, H)$ be a contact Hamiltonian system with contact Hamiltonian vector field $X_H$.
Consider a section $\gamma$ of $\pi_Q:  \cT Q\times \R\to  Q$ such that $\Ima \gamma$ is a Legendrian submanifold of $(\cT Q\times \R, \eta)$. Then, $X_H^\gamma$ and $X_H$ are $\gamma$-related if and only if 
\begin{equation}
    H\circ \gamma = 0\,. \label{eq:HJ_action_indep_autonomous}
\end{equation}
\end{thm}

The problem with this approach is that it cannot be used to completely integrate the system. Indeed, equation~\eqref{eq:HJ_action_indep_autonomous} implies that every integral curve of $X_H\circ \gamma$ is contained in $H^{-1}(0)$.
This can be solved by regarding the contact Hamiltonian system $(\cT Q\times \R, \eta, H)$ as the cocontact Hamiltonian system $(\R \times \cT Q\times \R, \d t, \eta, {\widehat H})$, where {$\widehat H= H \circ \rho_2$ (i.e. $\widehat H (t, q, p, z) = H(q,p, z)$),}  such that $\Rt({\widehat H})=0$ and making use of Theorem \ref{thm_action_indep}.
Suppose that $S$ is of the form $S(t, q) = \alpha(q) + \beta(t)$. Then, equation~\eqref{eq:HJ_Sardonian} yields
\begin{equation}
    H \circ j^1 \alpha + \frac{\partial\beta}{\partial t} = 0\, ,
\end{equation}
that is,
\begin{equation}
    H \left(q^i, \frac{\partial \alpha}{\partial q^i}, z\right) + \dot \beta(t) = 0\, .
    \label{HJ_contact_as_cocontact}
\end{equation}
With a suitable choice of $\alpha$ and $\beta$, one can cover energy levels distinct from $H=0$.

\begin{dfn}\label{def:complete_sols_indep_autonomous}
Let $(\T^*Q\times \R, \eta, H)$ be a contact Hamiltonian system, and let $(\R\times \cT Q\times \R, \tau, \eta, \widehat H = H\circ \rho_2)$ be its associated cocontact Hamiltonian system.
A \emph{complete solution of the action-independent Hamilton--Jacobi problem} for $(\T^*Q\times \R, \eta, H)$ is a map $\widehat\Phi\colon \R\times Q \times \R^{n} \to \R \times \cT Q \times \R$ such that $\Phi = \rho_2\circ \widehat \Phi$ is a local diffeomorphism and, for each $\lambda \in \R^{n}$,
\begin{equation}
\begin{aligned}
    \widehat \Phi_\lambda \colon \R\times Q &\to \R \times\cT Q \times \R  \\
    \left(t, q^i\right) &\mapsto \widehat\Phi\left(t, q^i, \lambda\right) 
\end{aligned}    
\end{equation}
is a solution of the action-independent Hamilton--Jacobi problem for $(\R\times \cT Q\times \R, \tau, \eta, \widehat H)$.
\end{dfn}
Let $\alpha\colon \R\times Q\times \R^{n} \rightarrow \R^{n}$, and $\pi_i\colon \R^{n} \to \R$ denote the canonical projections. One can define the $n$ functions $f_i=\pi_i\circ \alpha\circ \Phi^{-1}$ on $\R \times \cT Q \times \R$, so that the following diagram commutes:
\begin{equation}
        \begin{tikzcd}
        \R \times Q \times \R^{n} \arrow[r, "\Phi", shift left] \arrow[d, "\alpha"] &  \cT Q \times \R \arrow[l, "\Phi^{-1}", shift left] \arrow[d, "f_i"] \\
        \R^{n} \arrow[r, "\pi_i"] & \R
    \end{tikzcd}
\end{equation}
Then,
\begin{equation}
    \Ima \Phi_\lambda = \bigcap_{i=1}^n f_i^{-1} (\lambda_i)\, ,
\end{equation}
where $\Phi_\lambda(t,q) = \Phi(t,q,\lambda)$,
and
\begin{equation}
    \Ima \widehat \Phi_\lambda = \bigcap_{i=1}^n (f_i\circ \rho_2)^{-1} (\lambda_i)\, ,
\end{equation}
so the functions $f_i\circ \rho_2$ are constants of the motion for $\widehat H$, and thus the functions $f_i$ are constants of the motion for $H$.
\begin{exmpl}[The free particle with a linear external force]
    Consider the cocontact Hamiltonian system $(\R\times \cT Q \times \R, \d t, \eta, H)$, where
    \begin{equation}
        H = \frac{p^2}{2} - \kappa z,
    \end{equation}
    with $\kappa$ a positive constant. Let $\widehat H = H \circ \rho_2$ be the associated time-dependent Hamiltonian.
    Then, the action-independent Hamilton--Jacobi equation for $\widehat H$ is given by 
    \begin{equation}
        \frac{1}{2}\left( \frac{\partial S}{\partial q}\right)^2 - \kappa S(t,q) + \frac{\partial S}{\partial t} = 0\, ,
    \end{equation}
    that is,
     \begin{equation}
       \left( \frac{\partial S}{\partial q}\right)^2 - 2 \kappa S(t,q) + 2 \frac{\partial S}{\partial t} = 0\, ,
        \label{HJ_example_indep_auton}
    \end{equation}
    Suppose that the generating function $S$ is separable, namely, $S(t,q)= \alpha(t)+\beta(q)$. Then, equation~\eqref{HJ_example_indep_auton} can be written as
    \begin{equation}
         \left( \frac{\d \beta}{\d q}\right)^2 - 2 \kappa \alpha(t) - 2 \kappa \beta(q) + 2\frac{\d \alpha}{\d t} = 0\, ,
    \end{equation}
    so
   \begin{align}
       &  2\frac{\d \alpha}{\d t} -2 \kappa \alpha(t) = 0\,,\\
       & \left( \frac{\d \beta}{\d q}\right)^2 - 2 \kappa \beta(q) = 0.
   \end{align}
   Thus,
   \begin{equation}
       \alpha(t) = e^{\kappa t}\,, \qquad
       \beta_{\lambda}(q) = \left( \sqrt{\frac{\kappa}{2}} q + \lambda \right )^2,
   \end{equation}
   that is,
   \begin{equation}
       S_\lambda(t,q) = e^{\kappa t} + \left( \sqrt{\frac{\kappa}{2}} q + \lambda \right )^2,
   \end{equation}
   and 
   $$\widehat\Phi(t,q,\lambda) = j_t^1 S_\lambda(t,q) = \left(t, q, \sqrt{2\kappa}\left( \sqrt{\frac{\kappa}{2}} q + \lambda \right),  e^{\kappa t} + \left( \sqrt{\frac{\kappa}{2}} q + \lambda \right )^2\right)$$ 
   is a complete solution. 
   Then,
   $$\Phi\colon (t, q,\lambda) \mapsto \left(q, \sqrt{2\kappa}\left( \sqrt{\frac{\kappa}{2}} q + \lambda \right),  e^{\kappa t} + \left( \sqrt{\frac{\kappa}{2}} q + \lambda \right )^2\right)\, ,$$ 
   whose inverse is given by
   \begin{equation}
       \Phi^{-1} \colon (q, p, z) \mapsto
       \left(\frac{1}{\kappa}\log \abs{z-\frac{p^2}{2\kappa}}, q, 
       \frac{p-\kappa q}{\sqrt{2\kappa}}
       \right)\, .
   \end{equation}
   Hence,
   \begin{equation}
       f_1(t,q,p,z) =  \frac{p-\kappa q}{\sqrt{2\kappa}}
   \end{equation}
   is a conserved quantity.

  {The Hamiltonian vector field of $H$ is given by
   \begin{equation}
       X_H  = p\parder{}{q} + \kappa p \parder{}{p} + \left(\frac{p^2}{2} + \kappa z\right)\parder{}{z}\,.
   \end{equation}
   One can check that $X_H(f_1)=0$. Moreover,
    \begin{equation}
       X_H^\gamma
     = \restr{ p\parder{}{q}}{\Ima \Phi_\lambda}
     = \sqrt{2\kappa}\left( \sqrt{\frac{\kappa}{2}} q + \lambda \right)\parder{}{q}
     \,,
   \end{equation}
   whose integral curves $\sigma(t)=(t, q(t))$ are given by
   \begin{equation}
       q(t) = c e^{\kappa  t} -\sqrt{\frac{2}{\kappa}}\lambda\, ,
   \end{equation}
   where $c$ is a constant.
   Then, the integral curves of $X_H$ along $\Ima \Phi_\lambda$ are given by $\Phi_\lambda \circ \sigma(t) = (q(t), p(t), z(t))$, where
   \begin{equation}
       p(t) = \sqrt{2\kappa}\left( \sqrt{\frac{\kappa}{2}} q(t) + \lambda \right)
            = \kappa\, c e^{\kappa  t}\textbf{}
       \, ,
   \end{equation}
   and
   \begin{equation}
       z(t) =  e^{\kappa t} + \left( \sqrt{\frac{\kappa}{2}} q(t) + \lambda \right)^2
      =  e^{\kappa t} +  \frac{\kappa}{2} c^2 e^{2\kappa  t} \, .
   \end{equation}
    }
\end{exmpl}

\section{The action-dependent approach}
\label{section_HJ_action_dep}

\subsection{Hamilton--Jacobi theory. The action-dependent approach}

In Section \ref{section_HJ_action_indep} we have introduced a Hamilton--Jacobi theory for time-dependent contact Hamiltonian systems. In particular, this approach was shown to be useful to deal with time-independent contact Hamiltonian systems, where time is used as a free parameter. Nevertheless, this approach has a couple of drawbacks. First, complete solutions depend on $n+1$ parameters, instead of the $n$ parameters that are required for symplectic Hamiltonian systems \cite{Carinena2006}.
Additionally, time-independent solutions only cover the zero-energy level. 

In order to solve these problems, in this section we propose an alternative approach, considering solutions of the Hamilton--Jacobi problem depending on the action variable $z$.
Let us consider a section $\gamma$ of the bundle $\pi_Q^{t,z}: \R\times \T^*Q\times \R\to \R\times Q \times \R$, locally given by
  \begin{equation}
    \begin{aligned}
        \gamma: \R\times Q\times \R &\longrightarrow  \R\times \T^*Q\times \R\\
        (t, x, z) &\longmapsto \left(t, x, \gamma_i(t,x,z), z\right) \,.
    \end{aligned}
  \end{equation}
As in the previous approach, assume that $X_H^\gamma$ and $X_H$ are $\gamma$-related, so that the following diagram commutes:  
\begin{center}
    \begin{tikzcd}
        \R \times \cT Q \times \R \arrow[d, "\pi_Q^{t,z}"] \arrow[r, "X_H"]                           & \T(\R \times \cT Q\times \R) \arrow[d, swap, "\T\pi_Q^{t,z}"]                       \\
        \R \times Q \times \R \arrow[r, "X_H^{\gamma}"] \arrow[u, "\gamma", bend left] & \T(\R \times Q \times \R) \arrow[u, swap, "\T \gamma", bend right]
    \end{tikzcd}
\end{center}

Locally,
\begin{equation}
    X_H \circ \gamma = \frac{\partial}{\partial t} 
    + \frac{\partial H}{\partial p_i} \frac{\partial }{\partial q^i}
    - \left( \frac{\partial H }{\partial q^i} + \gamma_i \frac{\partial H}{\partial z} \right) \frac{\partial }{\partial p_i}
    + \left(\gamma_i \frac{\partial H}{\partial p_i} -H\right) \frac{\partial }{\partial z}\,,
\end{equation}
and
\begin{equation}
   \T\gamma \circ X_H^\gamma = \frac{\partial}{\partial t} 
    + \frac{\partial H}{\partial p_i} \frac{\partial }{\partial q^i}
    + \left(\frac{\partial \gamma_i} {\partial t} + \frac{\partial H} {\partial p_j} \frac{\partial \gamma_i} {\partial q^j} 
    + \left(\gamma_j \frac{\partial H} {\partial p_j} - H  \right)\frac{\partial \gamma_i} {\partial z}  \right) \frac{\partial  } {\partial p_i}
    + \left(\gamma_i \frac{\partial H}{\partial p_i} -H\right) \frac{\partial }{\partial z}\, ,
\end{equation}
so $X_H^\gamma$ and $X_H$ are $\gamma$-related if and only if
\begin{equation}
   - \left( \frac{\partial H }{\partial q^i} + \gamma_i \frac{\partial H}{\partial z} \right)
   = \frac{\partial \gamma_i} {\partial t} + \frac{\partial H} {\partial p_j} \frac{\partial \gamma_i} {\partial q^j} 
    + \frac{\partial \gamma_i} {\partial z} \left(\gamma_j \frac{\partial H} {\partial p_j} - H  \right) \, .
    \label{pre_HJ_Lainzian}
\end{equation}


Note that $\Ima \gamma $ is $(n + 2)$-dimensional, so it no longer makes sense to require it to be Legendrian \cite{deLeon2022d}. We will require it to be coisotropic instead.


Let us recall that, given a Jacobi manifold $(M, \Lambda, E)$ and a distribution $\mathcal D\subseteq \T M$, the \emph{orthogonal complement} $\mathcal{D}^\perp$ of $\mathcal D$ is given by \cite{deLeon2021a,Libermann1987}
     $$ \mathcal{D}_x^\perp = \hat\Lambda\left(\mathcal{D}_x^\circ\right)\,, $$
where $\mathcal{D}_x^\circ = \left\{\alpha_x \in \cT_x M \mid \alpha_x(v)=0,\ \forall\ v \in \mathcal{D}_x\right\}$ denotes the annihilator. In particular, a cocontact manifold $(M,\tau,\eta)$ is a Jacobi manifold (see Proposition \ref{prop_Jacobi_mfold}) and its morphism $\hat \Lambda$ is given by equation~\eqref{Lambda_hat_cocontact}. A submanifold $N\hookrightarrow M$ is said to be \emph{coisotropic} if $\T N^\perp \subseteq \T N$.

\begin{dfn}
   Given a section $\alpha: \R \times Q \times \mathbb{R} \to \R \times \bigwedge\nolimits^k\cT Q \times \mathbb{R}$ and $t, z \in \mathbb{R}$, let 
   \begin{equation}
       \begin{aligned}
           \alpha_{(t,z)} :  Q &\longrightarrow \bigwedge\nolimits^k\cT Q\\
           x &\longmapsto \pr_{\Lambda^k \cT Q} (\alpha(t,x,z))\, ,
       \end{aligned}
   \end{equation}
   where $\pr_{\bigwedge\nolimits^k \cT Q}: \R \times \bigwedge\nolimits^k\cT Q \times \mathbb{R} \to \bigwedge\nolimits^k\cT Q$ is the canonical projection. The \emph{exterior derivative of $\alpha$ at fixed $t$ and $z$} is the section of $\R \times \bigwedge\nolimits^{k+1}\cT Q \times \mathbb{R} \to \R \times Q \times \mathbb{R}$ given by
   \begin{equation}
       \d_Q \alpha(t, x, z) = (t, \d \alpha_{(t,z)} (x), z)\,.
   \end{equation}
\end{dfn}

The coisotropic condition can be written in local coordinates as follows.

\begin{lem}
    Assume that an $(n+2)$-dimensional submanifold $N$ of a $(2n+2)$-dimensional cocontact manifold $(M, \tau, \eta)$ is locally the zero set of the constraint functions $\left\{\phi_a\right\}_{a=1,\dotsc,n}$.
    Then, $N$ is coisotropic if and only if the following equation holds in Darboux coordinates:
     \begin{equation}\label{eq:coisotropic_coords}
        \left(\parder{\phi_a}{q^i} + p_i \parder{\phi_a}{z}\right)\parder{\phi_b}{p_i} - 
         \left(\parder{\phi_b}{q^i} + p_i\parder{\phi_b}{z}\right)\parder{\phi_a}{p_i} = 0\, .
    \end{equation}
\end{lem}
\begin{proof}
     Assume that $(M,\tau, \eta)$ is a $(2n+2)$-dimensional cocontact manifold. Let $N\hookrightarrow M$ be a $k$-dimensional submanifold locally given as the zero set of functions $\phi_a:U\to \R$, with $a\in \{1, \ldots, 2n+2-k\}$. We have that
    \begin{equation*}
        {\T N}^{\perp} = \left\langle\left\{Z_a\right\}_{a=1,\dotsc,2n+2-k}\right\rangle\, ,
    \end{equation*}
    where
    \begin{equation*}
        Z_a = \hat \Lambda(\d \phi_a)
        = \left(\parder{\phi_a}{q^i} + p_i \parder{\phi_a}{z}\right)\parder{}{p_i} - 
        \parder{\phi_a}{p_i} \left(\parder{}{q^i} + p_i\parder{}{z}\right)\,.
    \end{equation*}    
    Therefore, $N$ is coisotropic if and only if $Z_a(\phi_b)=0$ for all  $a,b$, which in Darboux coordinates yields 
    equation~\eqref{eq:coisotropic_coords}.
    
\end{proof}

\begin{prop}
    Let $\gamma$ be a section of $\R\times \cT Q \times \mathbb{R}$ over $\R \times Q \times \mathbb{R}$. Then $\Ima \gamma$ is a coisotropic submanifold if and only if
    \begin{equation}
        \frac{\partial \gamma_i} {\partial q^j} + \gamma_j \frac{\partial \gamma_i} {\partial z} 
        = \frac{\partial \gamma_j} {\partial q^i} + \gamma_i \frac{\partial \gamma_j} {\partial z}.
        \label{eq_coisotropic_coords}
    \end{equation}
\end{prop}
\begin{proof}
    Equation~\eqref{eq_coisotropic_coords} is obtained by applying the previous result to the submanifold $N = \Ima \gamma$, which is locally defined by the constraints $\phi_i = p_i - \gamma_i$.
\end{proof}



    


Now suppose that the $\gamma$ appearing in equation~\eqref{pre_HJ_Lainzian} is such that $\Ima \gamma$ is coisotropic. Then, by means of equation~\eqref{eq_coisotropic_coords} we obtain
\begin{equation}
    \frac{\partial H }{\partial q^i}  + \frac{\partial H} {\partial p_j} \frac{\partial \gamma_j} {\partial q^i} 
    + \gamma_i \left(\frac{\partial H} {\partial p_j} \frac{\partial \gamma_j} {\partial z}  +\frac{\partial H}{\partial z} \right)
     + \frac{\partial \gamma_i} {\partial t}
    = H \frac{\partial \gamma_i} {\partial z} \, ,
\end{equation}
or, globally,
\begin{equation}
   \d_Q \left(H\circ \gamma  \right) + \frac{\partial}{\partial z} (H\circ \gamma) \gamma + \liedv{\Rt} \gamma = (H\circ \gamma) \liedv{\frac{\partial}{\partial z}} \gamma\, .
   \label{eq:HJ_Lainzian}
\end{equation}

\begin{thm}[Action-dependent Hamilton--Jacobi Theorem]
Let $\gamma$ be a section of $\pi_Q^{t,z}: \R\times \T^*Q\times \R\to \R\times Q\times \R$ such that $\Ima \gamma$ is a {coisotropic} submanifold of $(\R\times \T^*Q\times \R, \tau, \eta)$. Then, $X_H^\gamma$ and $X_H$ are $\gamma$-related if and only if equation~\eqref{eq:HJ_Lainzian} holds. This equation will be called the \emph{action-dependent Hamilton--Jacobi} equation for $(\R\times \T^*Q\times \R, \tau, \eta, H)$.
\end{thm}

\begin{dfn}
Let $(\R\times \T^*Q\times \R, \tau, \eta, H)$ be a cocontact Hamiltonian system.
A \emph{complete solution of the action-dependent Hamilton--Jacobi problem} for $(\R\times \T^*Q\times \R, \tau, \eta, H)$ is a local diffeomorphism $\Phi\colon \R\times Q \times \R^n \times \R\to \R \times \cT Q \times \R$ such that, for each $\lambda \in \R^n$,
\begin{equation}
\begin{aligned}
    \Phi_\lambda \colon \R\times Q \times \R&\longrightarrow \R \times\cT Q \times \R  \\
    \left(t, q^i, z\right) &\longmapsto \Phi\left(t, q^i, \lambda, z\right) 
\end{aligned}    
\end{equation}
is a solution of the action-dependent Hamilton--Jacobi problem for $(\R\times \T^*Q\times \R, \tau, \eta, H)$.
\end{dfn}

Let $\alpha\colon \R\times Q\times \mathbb{R}^n\times \R\rightarrow \mathbb{R}^n$, and $\pi_i\colon \R^n \to \R$ denote the canonical projections. Let us define the functions $f_i=\pi_i\circ \alpha\circ \Phi^{-1}$ on $\R \times \cT Q \times \R$, so that the following diagram commutes:
\[
    \begin{tikzcd}
        \R \times Q \times \R^n \times \R \arrow[r, "\Phi", shift left] \arrow[d, "\alpha"] & \R\times \cT Q \times \R \arrow[l, "\Phi^{-1}", shift left] \arrow[d, "f_i"] \\
        \R^n \arrow[r, "\pi_i"] & \R
    \end{tikzcd}
\]

\begin{thm}\label{thm:complete_solution_involution}
   Let $\Phi\colon \R\times Q \times \R^n \times \R\to \R \times \cT Q \times \R$ be a complete solution of the action-dependent Hamilton--Jacobi problem for $(\R\times \T^*Q\times \R, \tau, \eta, H)$. Then, 
    \begin{enumerate}[{\rm (i)}]
        \item For each $i\in \{1, \dotsc, n\}$, the function $f_i=\pi_i\circ \alpha\circ \Phi^{-1}$ is a constant of the motion. However, these functions are not necessarily in involution, i.e., $\{f_i, f_j\}\neq 0$.
        \item For each $i\in \{1, \dotsc, n\}$, the function $\hat f_i= g f_i$, where $g$ is a dissipated quantity, is also a dissipated quantity. Moreover, if $\Rt H = 0$ and taking $g=H$, these functions are in involution, i.e., $\{\hat f_i, \hat f_j\}= 0$.
    \end{enumerate}
\end{thm}

\begin{proof}
 Observe that
$$
\Ima \Phi_\lambda = \bigcap_{i=1}^n \, f_i^{-1}(\lambda_i)\,,
$$
where $\lambda = (\lambda, \dotsc, \lambda_n) \in \mathbb{R}^n$. In other words,
$$
\Ima \Phi_\lambda = \{ x \in \R\times \cT Q \times \R \mid f_i(x) = \lambda_i, i=1, \dotsc, n\}\,.
$$
Therefore, since $X_H$ is tangent to any of the submanifolds $\Ima \Phi_\lambda$, we deduce that
$$
X_H (f_i) = 0\,.
$$
Moreover, we can compute
$$
\{f_i, f_j\} = \Lambda (\d f_i, \d f_j) - f_i \Rz (f_j) + f_j \Rz (f_i)\,,
$$
but
$$
\Lambda (\d f_i, \d f_j) = \hat\Lambda(\d f_i)(f_j) = 0\,,
$$
since $(\T \Ima \Phi_\lambda)^\perp = \hat\Lambda ((\T \Ima \Phi_\lambda)^\circ) \subset \T \Ima \Phi_\lambda$, so 
\begin{equation}
\{f_i, f_j\} = - f_i \Rz(f_j) + f_j \Rz(f_i)\,.
\end{equation}

By Proposition \ref{prop:conserved-dissipated}, we already know that the product of a conserved quantity and a dissipated quantity is a dissipated quantity. Let $f_i$ and $f_j$ be conserved quantities and take $g = H$. Then, by equation~\eqref{eq_involution_dissipated}, $\{\hat f_i, \hat f_j \}$ vanishes.

\end{proof}

From a complete solution of the Hamilton--Jacobi problem one can reconstruct the dynamics of the system. If $\sigma$ is an integral curve of the vector field $X_H^\gamma$, then $\Phi_\lambda \circ \sigma$ is an integral curve of $X_H$, thus recovering the dynamics of the original system.

\subsection{Integrable contact Hamiltonian systems}
Let $(\T^*Q\times \R, \eta, H)$ be a contact Hamiltonian system. Recall that 
the action-dependent Hamilton--Jacobi equation for $(\T^*Q\times \R, \eta, H)$ is given by \cite{deLeon2021d}
\begin{equation}
   \d_Q \left(H\circ \gamma  \right) + \frac{\partial}{\partial z} (H\circ \gamma) \gamma = (H\circ \gamma) \liedv{\frac{\partial}{\partial z}} \gamma\, .
   \label{eq:HJ_Lainzian_autonomous}
\end{equation}
A \emph{complete solution of the action-dependent Hamilton--Jacobi problem} for $( \T^*Q\times \R, \eta, H)$ is a local diffeomorphism $\Phi\colon Q \times \R^n \times \R\to \cT Q \times \R$ such that, for each $\lambda \in \R^n$,
\begin{equation}
\begin{aligned}
    \Phi_\lambda \colon Q \times \R&\longrightarrow \cT Q \times \R  \\
    \left(q^i, z\right) &\longmapsto \Phi\left(q^i, \lambda, z\right) 
\end{aligned}    
\end{equation}
is a solution of the action-dependent Hamilton--Jacobi problem for $( \T^*Q\times \R, \eta, H)$.

Let $\Phi\colon Q \times \R^n \times \R\to \cT Q \times \R$ be a complete solution of the Hamilton--Jacobi problem for $(\T^*Q\times \R, \eta, H)$. Then,
\begin{equation}
    \mathcal{F} = \left\{\mathcal{F}_\lambda =\Ima \Phi_\lambda \mid \lambda\in \R^n \right\}
    \subseteq \cT Q\times \R
\end{equation}
is a foliation in coisotropic submanifolds.

In the symplectic case, since solutions of the Hamilton--Jacobi equation are closed one-forms on $Q$, the images of a complete solution for each choice of parameters $\lambda$ form a Lagrangian foliation invariant under the action of the Hamiltonian flow. This structure is called an integrable system. In analogy, we introduce the following definition:
\begin{dfn} \label{def:integrable_contact}
Let $(M, \eta, H)$ be a contact Hamiltonian system and let $\mathcal{F}$ be a foliation consisting of $(n+1)$-dimensional coisotropic (with respect to the Jacobi structure of the contact manifold) leaves invariant under the flow of the Hamiltonian vector field $X_H$. Then we call  $(M, \eta, H, \mathcal{F})$ an \emph{integrable system}.
\end{dfn}

\begin{rmrk}\label{remark_LA_2}
    Each of the leaves $\mathcal F_\lambda$ is invariant under the flow of $X_H$ and $X_{f_i}$.
    Since $\mathcal F_\lambda$ is an $(n+1)$-dimensional manifold with $n+1$ independent and commuting tangent vector fields, if the vector fields are complete, by \cite[Ch.\;10, Lem.\;2]{Arnold1978} it is diffeomorphic to $\mathbb{T}^k\times \R^{n+1-k}$, where $\mathbb{T}^k$ is the $k$-dimensional torus.
\end{rmrk}

The definition above can be compared to the ones given in~\cite{Boyer2011, Khesin2010}:
\begin{itemize}
    \item In~\cite{Boyer2011}, Boyer proposes a concept of completely integrable system for the so-called good Hamiltonians, that is, the Hamiltonian function is preserved along the flow of the Reeb vector field. This is a particular case of our definition, in which both the Hamiltonian and the constants of the motion do not depend on $z$.
    \item In~\cite{Khesin2010}, Khesin and Tabachnikov call a foliation \emph{co-Legendrian} when it is transverse to $\mathcal{H}$ and $\T\mathcal{F}\cap \mathcal{H}$ is integrable. Then they define an integrable system as a particular case of a co-Legendrian foliation with some extra regularity conditions. In the case that the dimension of the leaves is $n+1$, the following proposition shows that co-Legendrian foliations are particular cases of coisotropic foliations.
\end{itemize}

\begin{prop}    
    Let $i\colon N\hookrightarrow M$ be a submanifold of a $(2n+1)$-dimensional contact manifold $(M, \eta)$.  If $N$ is an $(n+1)$-dimensional co-Legendrian submanifold, then it is also a coisotropic submanifold.
\end{prop}

\begin{proof}
    Let us write $\T N = \mathcal{D}_{\mathcal{H}} \oplus \mathcal{E}$, where $\mathcal{D}_{\mathcal{H}}=\T N \cap \mathcal{H}$. Then, $\T N^\perp =  \mathcal{D}_{\mathcal{H}}^\perp \cap \mathcal{E}^\perp$. Obviously, $\eta$ vanishes in $\T N \cap \mathcal{H}$. Moreover, since $\mathcal{D}_{\mathcal{H}}$ is integrable,  
    \begin{equation}
    0 = \eta([v, w]) = \contr{[v,w]} \eta 
    = \liedv{v} \contr{w} \eta - \contr{w} \liedv{v} \eta 
    = -\contr{w}  \contr{v} \d \eta - \contr{w} \d \contr{v} \eta
    =  -\contr{w}  \contr{v} \d \eta \, , 
\end{equation} 
for any $v, w \in \mathcal{D}_{\mathcal{H}}$, so $\d \eta_{\mid \mathcal{D}_{\mathcal{H}}}=0$. Observe that $\hat \Lambda_{\mid \mathcal{H}} = \sharp_{\mid \mathcal{H} }$, and $\sharp_{\mid \mathcal{H}}\colon \mathcal{H} \to \langle\Reeb\rangle^\circ $, $\sharp_{\mid \mathcal{H}}^{-1}(v)=\contr{v} \d \eta$ is an isomorphism. Since  $\d \eta_{\mid \mathcal{D}_{\mathcal{H}}}=0$, $\sharp_{\mid \mathcal{H}}^{-1}(\mathcal{D}_{\mathcal{H}}) \subseteq \mathcal{D}_{\mathcal{H}}^\circ$. Thus, $\mathcal{D}_{\mathcal{H}} \subseteq  \hat{\Lambda} (\mathcal{D}_{\mathcal{H}}^\circ)=\mathcal{D}_{\mathcal{H}}^\perp$. By a dimension counting argument, we can see that both spaces are equal and, thus, $\mathcal{D}_{\mathcal{H}} = \mathcal{D}_{\mathcal{H}}^\perp$.


\end{proof}

We also note that a foliation $\tilde{\mathcal{F}}$ by Legendrian submanifolds can never be invariant by the Hamiltonian flow. Indeed, let $\tilde{F} \in \tilde{\mathcal{F}}$. The leaves of $\tilde{\mathcal{F}}$ are Lagrangian, thus $\T\tilde{F}_0 \subseteq \ker \eta$. Since $\eta(X_H) = - H$, $X_H$ can only be tangent to the leaves at the zero set of $H$, hence its flow cannot leave invariant the whole foliation. 

Observe that Definition \ref{def:integrable_contact} can be naturally extended to cocontact Hamiltonian systems.

\begin{dfn}
Let $(M, \tau, \eta, H)$ be a cocontact Hamiltonian system and let $\mathcal{F}$ be a foliation consisting of $(n+2)$-dimensional coisotropic leaves (with respect to the Jacobi structure of the cocontact manifold) invariant under the flow of the cocontact Hamiltonian vector field $X_H$. Then we call  $(M,  \tau, \eta, H, \mathcal{F})$ an \emph{integrable cocontact system}.
\end{dfn}




\subsection{Example 1: freely falling particle with linear dissipation}

Consider a particle of time-dependent mass $m(t)$ which is freely falling and subject to a dissipation linear in the velocity with proportionality constant $\gamma$. The Hamiltonian function $H:\R\times \T^*\R \times \R\to \R$ is given by
\begin{equation}
    H(t, q,p,z) = \frac{p^2}{2m(t)} + m(t)g q + \frac{\gamma}{m(t)} z\,,
\end{equation}
where $g$ is the gravity. The Hamiltonian vector field corresponding to this Hamiltonian function is
$$ X_H = \parder{}{t} + \frac{p}{m(t)}\parder{}{q} - \left( m(t)g + \frac{\gamma}{m(t)} p \right)\parder{}{p} + \left(\frac{p^2}{2m(t)} - m(t)g q - \frac{\gamma}{m(t)} z\right)\parder{}{z}\,. $$
Its integral curves $(t(r), q(r), p(r), z(r))$ satisfy the system of differential equations
$$ \dot t = 1\,,\qquad \dot q = \frac{p}{m(t)}\,,\qquad \dot p = -m(t)g - \frac{\gamma}{m(t)}p\,,\qquad \dot z = \frac{p^2}{2m(t)} - m(t)g q - \frac{\gamma}{m(t)} z\,. $$
Combining the second and third equations, we get
$$ \frac{\d}{\d t}(m(t)\dot q) = -m(t) g - \gamma\dot q\,. $$

In order to solve the Hamilton--Jacobi problem, we look for a conserved quantity linearly independent from the Hamiltonian, i.e., a function $f$ on $\T\R\times \R$ such that $X_H f = 0$. For the sake of simplicity, one can assume that $f$ does not depend on $q$ or $z$.
Indeed, one can verify that
$$f(t,q,p,z)=e^{\int _1^t\frac{\gamma }{m(s)}\d s} \left(p + g e^{-\int _1^t\frac{\gamma }{m(s)}\d s} \int _1^te^{\int _1^u\frac{\gamma}{m(s)}\d s}  m(u)\d u\right)
$$ 
is a conserved quantity. We can thus express the momentum $p$ as a function of $t$ and a real parameter $\lambda$, namely, 
$$P(t,\lambda)=e^{-\int _1^t\frac{\gamma }{m(s)}\d s}\left(\lambda - ge^{-\int _1^t\frac{\gamma }{m(s)}\d s} \int _1^t e^{\int _1^u\frac{\gamma }{m(s)}\d s} m(u)\d u\right)\, ,$$
and obtain a complete solution of the Hamilton--Jacobi problem for $H$:
\begin{equation}
    \phi_\lambda:(t, q, z)\longmapsto \left(t, q,p \equiv e^{-\int _1^t\frac{\gamma }{m(s)}\d s}\left(\lambda - ge^{-\int _1^t\frac{\gamma }{m(s)}\d s} \int _1^t e^{\int _1^u\frac{\gamma }{m(s)}\d s} m(u)\d u\right), z \right)\,.
\end{equation}
In this case, equation \eqref{eq_coisotropic_coords} holds trivially, so $\Ima \Phi_\lambda$ is coisotropic.

In addition, one can verify that
$$ k(t,q,p,z)=p + ge^{-\int_1^t\frac{\gamma }{m(s)}\d s} \int_1^t e^{\int _1^u\frac{\gamma }{m(s)}\d s} m(u)\d u
=e^{-\int_1^t\frac{\gamma }{m(s)}\d s} f(t,q,p,z)$$
is a dissipated quantity, that is, $\left\{k, H\right\} - \Rt k = 0$.


\subsection{Example 2: damped forced harmonic oscillator}

Consider the product manifold $\R\times\cT \R\times\R$ with natural coordinates $(t, q, p, z)$. The Hamiltonian function
$$ H(t,q,p,z) = \frac{p^2}{2m} + \frac{k}{2}q^2 - qF(t) + \frac{\gamma}{m} z $$
describes a harmonic oscillator with elastic constant $k$, friction coefficient $\gamma$ and subjected to an external time-dependent force $F(t)$ \cite{deLeon2022d}.

The Hamiltonian vector field is
\begin{equation*}
    X_H = \parder{}{t} + \frac{p}{m}\parder{}{q} + \left( -kq + F(t) - \frac{p}{m}\gamma \right)\parder{}{p} + \left( \frac{p^2}{2m} - \frac{k}{2}q^2 + qF(t) - \frac{\gamma}{m} z \right)\parder{}{z}\,,
\end{equation*}
and its integral curves $(t(r),q(r),p(r),z(r))$ satisfy
\begin{equation*}
        \dot t = 1\,,\qquad \dot q = \frac{p}{m}\,,\qquad \dot p = -kq + F(t) - \frac{p}{m}\gamma\,,\qquad \dot z = \frac{p^2}{2m} - \frac{k}{2}q^2 + qF(t) - \frac{\gamma}{m} z\,.
\end{equation*}
Combining the second and the third equations above, we obtain the second-order differential equation
$$ m\ddot q + \gamma \dot q + kq = F(t) \,, $$
which corresponds to a damped forced harmonic oscillator. One can check that the function
\begin{multline}
    g(t,q,p,z) =  e^{\frac{\gamma  t}{2 m}} \left(\frac{\sinh \left(\frac{\kappa t}{2 m}\right) (2 k m q+\gamma  p)}{\kappa}+p \cosh \left(\frac{\kappa t}{2 m}\right)\right)\\
    -\int _1^t F(s) e^{\frac{\gamma  s}{2 m}} \left(\cosh \left(\frac{\kappa s}{2 m}\right)+\frac{\gamma  \sinh \left(\frac{\kappa s}{2 m}\right)}{\kappa}\right)\d s\, ,
\end{multline}
where $\kappa = \sqrt{\gamma^2 - 4km}$, is a conserved  quantity. 
It is worth noting that, since $\sinh = x +\mathcal{O}(x^3)$ and $\cosh x=1+\mathcal{O}(x^2)$ near $x=0$, $\sinh(ix)=i\sin x$ and $\cosh (ix)=\cos x$, the equation above is well-defined and real-valued for any of $\kappa\in \mathbb{C}$.
Thus, we can write $p$ in terms of $t,q,z$ and a real parameter $\lambda$ as
\begin{align*}
    P(t,q, \lambda, z)
    =\frac{e^{-\frac{\gamma  t}{2 m}} \left(\kappa \int _1^te^{\frac{s \gamma }{2 m}} F(s) \left(\cosh \left(\frac{\kappa s}{2 m}\right)+\frac{\gamma  \sinh \left(\frac{\kappa s}{2 m}\right)}{\kappa}\right)\d s-2 k m q e^{\frac{\gamma  t}{2 m}} \sinh \left(\frac{\kappa t}{2 m}\right)+\kappa \lambda \right)}{\gamma  \sinh \left(\frac{\kappa t}{2 m}\right)+\kappa \cosh \left(\frac{\kappa t}{2 m}\right)}\, ,
\end{align*}
and obtain a complete solution of the Hamilton--Jacobi problem:
\begin{equation}
    \Phi_\lambda\colon (t, q, \lambda, z) \mapsto \left(t, q, p\equiv P\left(t, q,\lambda, z\right), z \right)\, .
\end{equation}
Obviously equation \eqref{eq_coisotropic_coords} is satisfied, hence $\Ima \Phi_\lambda$ is coisotropic. In addition, 
\begin{multline*}
    f(t,q,p,z) = e^{-\frac{\gamma  t}{m}} \left[e^{\frac{\gamma  t}{2 m}} \left(\frac{\sinh \left(\frac{\kappa t}{2 m}\right) (2 k m q+\gamma  p)}{\kappa}+p \cosh \left(\frac{\kappa t}{2 m}\right)\right)
    \right.\\\quad\left.
    -\int _1^te^{\frac{s \gamma }{2 m}} F(s) \left(\cosh \left(\frac{\kappa s}{2 m}\right)+\frac{\gamma  \sinh \left(\frac{\kappa s}{2 m}\right)}{\kappa}\right)\d s\right]\,,
\end{multline*}
is a dissipated quantity.

\section{Conclusions and outlook} \label{section_conclusions}

The main contributions of the present paper are the following:

\begin{itemize}
    \item We have obtained two different Hamilton–Jacobi equations for time-dependent contact Hamiltonian systems: the so-called action independent and action-dependent approaches. In particular, the action-independent approach is useful for time-independent contact Hamiltonian systems, where the use of time as a free parameter allows to integrate the system at non-zero energy levels. In addition, we have introduced a notion of complete solution in the action-independent approach. 
    Each of these complete solutions is associated with a family of $n+1$ independent dissipated quantities in involution (where $n$ is the number of degrees of freedom of the system). 
    \item The action-dependent approach also permits to introduce a natural notion of complete solution to the Hamilton--Jacobi problem. Each of these complete solutions is associated with a family of $n$ independent dissipated quantities in involution. Moreover, the image of a complete solution is a coisotropic submanifold.
    \item We introduce a new notion of integrable system in a contact manifold, taking into account the dynamics given by the Hamiltonian vector field, and extending the concept of complete solution. This allows us to study the dynamics outside the zero-energy level.
\end{itemize}

As we have pointed out in Remarks~\ref{remark_LA} and \ref{remark_LA_2}, there is a relationship between solutions of Hamilton--Jacobi equations and several notions of integrability. Namely, the existence of foliations by coisotropic tori, integrability by quadratures and the construction of action-angle coordinates. Further research is needed to clarify these notions and their relationships in contact Hamiltonian systems.

Other topics for future research include the reduction problem, the Hamilton--Jacobi equations for the evolution vector field and its possible applications to thermodynamics as well as the extension to higher order systems. The study of the discrete Hamilton--Jacobi equations and applications to the construction of geometric integrators is also on the agenda.








\section*{Declaration of interests}
The authors have no conflicts to declare.

\section*{Acknowlegments}
We are thankful to Miguel C. Muñoz-Lecanda for his enriching comments on a previous version of the preprint. {We thank the referee for his/her helpful comments, which have been quite helpful in improving the clarity and overall quality of the paper.} M.~d.~L., M.~L.~and A.~L.-G.~acknowledge financial support of the Spanish Ministry of Science and Innovation (MCIN/AEI/ 10.13039/501100011033), under grants PID2019-106715GB-C2 and ``Severo Ochoa Programme for Centres of Excellence in R\&D'' (CEX2019-000904-S). M.~d.~L.~ also acknowledges Grant EIN2020-112197, funded by AEI/10.13039/501100011033 and European Union 
NextGenerationEU/PRTR. M.~L.~wishes to thank MCIN for the predoctoral contract PRE2018-083203. A.~L.-G. would also like to thank MCIN for the predoctoral contract PRE2020-093814. X.~R.~acknowledges financial support of the Ministerio de Ciencia, Innovación y Universidades (Spain), project PGC2018-098265-B-C33.


\let\emph\oldemph
\printbibliography

\end{document}